\documentclass[11pt]{article}

\usepackage{amsmath}
\usepackage{amssymb}
\usepackage{amsthm}
\usepackage{tikz}
\usepackage{graphicx}
\usepackage{float}
\usepackage{algorithm}
\usepackage[noend]{algpseudocode}
\usepackage{tabularx}
\usepackage[margin=1in]{geometry}
\usepackage{arydshln}
\usepackage{multirow}
\usepackage{wrapfig}

\usetikzlibrary{positioning,calc,shapes,arrows}
\usetikzlibrary{backgrounds}
\usetikzlibrary{matrix,shadows,arrows} 
\usetikzlibrary{decorations.pathreplacing}

\def\etal.{et\penalty50\ al.}
\theoremstyle{plain}
\newtheorem{theorem}{Theorem}[section]
\newtheorem{lemma}[theorem]{Lemma}

\theoremstyle{definition}
\newtheorem{definition}{Definition}[section]

\theoremstyle{remark}
\newtheorem{question}{Question}[section]

\newtheorem{remark}[question]{Remark}

\theoremstyle{plain}
\newtheorem*{theorem*}{Theorem}

\newcommand{\poly}{\ensuremath{\mathsf{poly}}}

\title{A Simple PTAS for the Dual Bin Packing Problem and Advice Complexity of Its Online Version}

\author{Allan Borodin\thanks{Research is supported by NSERC.} \\ University of Toronto \\ \textsf{bor@cs.toronto.edu}  \and
Denis Pankratov\footnotemark[1] \\ University of Toronto \\ \textsf{denisp@cs.toronto.edu} \and
Amirali Salehi-Abari\footnotemark[1] \\ University of Toronto \\ \textsf{abari@cs.toronto.edu}}

\begin{document}

\maketitle

\begin{abstract}
Recently, Renault (2016) studied the dual bin packing problem in the per-request advice model of online algorithms. He showed that given $O(1/\epsilon)$ advice bits for each input item allows approximating the dual bin packing problem online to within a factor of $1+\epsilon$. Renault asked about the advice complexity of dual bin packing in the tape-advice model of online algorithms. We make progress on this question. Let $s$ be the maximum bit size of an input item weight. We present a conceptually simple online algorithm that with total advice $O\left(\frac{s + \log n}{\epsilon^2}\right)$ approximates the dual bin packing to within a $1+\epsilon$ factor. To this end, we describe and analyze a simple offline PTAS for the dual bin packing problem. Although a PTAS for a more general problem was known prior to our work (Kellerer 1999, Chekuri and Khanna 2006), our PTAS is arguably simpler to state and analyze. As a result, we could easily adapt our PTAS to obtain the advice-complexity result.

We also consider whether the dependence on $s$ is necessary in our algorithm. We show that if $s$ is unrestricted then for small enough $\epsilon > 0$ obtaining a $1+\epsilon$ approximation to the dual bin packing requires $\Omega_\epsilon(n)$ bits of advice. To establish this lower bound we analyze an online reduction that preserves the advice complexity and approximation ratio from the binary separation problem due to Boyar et al. (2016). We define two natural advice complexity classes that capture the distinction similar to the  Turing machine world distinction between pseudo polynomial time algorithms and polynomial time algorithms. Our results on the dual bin packing problem imply the separation of the two classes in the advice complexity world.
\end{abstract}

\section{Introduction}
\label{sec:intro}
Given a sequence of items of weights $w_1,\ldots, w_n$ and $m$ bins of unit capacity, the dual bin packing problem asks for the maximum number of items that can be packed into the bins without exceeding the capacity of any bin.\footnote{This terminology is somewhat unfortunate, because the dual bin packing problem is not the dual to the natural integer programming formulation of the bin packing problem. For some early results on the latter see~\cite{assmann1984,csirik1988}.} The search version of this problem is to find a good packing. In the online version of this problem, the items are presented one at a time in some adversarial order and the algorithm needs to make an irrevocable decision into which (if any) bin to pack the current item. The dual bin packing problem has a substantial history in both the offline and online settings starting with Coffman et al. \cite{CoffmanLT78}. The performance of the online algorithm is measured by its competitive ratio; 
 that is, the worst-case ratio between the value of an offline optimal solution and the value of the solution obtained by the algorithm. It is known that the online dual packing problem does not admit a constant competitive ratio even for randomized algorithms \cite{BoyarLN01, CyganJS16}. The assumption that the online algorithm does not see the future at all is quite restrictive and in many cases impractical. It is often the case that some information about the input sequence is known in advance, e.g., its length, the largest weight of an item, etc. An information-theoretic way of capturing this side knowledge is given by the \emph{tape-advice model} \cite{bockenhauer2009advice}. In this model, an all powerful oracle that sees the entire input sequence creates a short advice string. The algorithm uses the advice string in processing the online nodes. The main object of interest here is the tradeoff between the size of advice and the competitive ratio of an online algorithm. Often, a short advice string results in a dramatic improvement of the best competitive ratio that is achievable by an online algorithm. Of course, a short advice string can be computationally difficult to obtain since the oracle is allowed unlimited power.

A related advice model is the \emph{per-request advice model} \cite{emek2011online}. In this model, prior to seeing the $i$th input item, the algorithm receives the $i$th advice string. Unlike the tape-advice model, the overall length of advice is always lower bounded by $n$ in this model. Both of these advice models have recently received considerable attention in the research community (see Boyer et al \cite{boyar2016online} for an extensive  survey on this topic). Recently, Renault \cite{renault2016online} studied the dual bin packing problem in the per-request advice model. He designed an algorithm that with $1$ bit of advice per request achieves a $3/2$ competitive ratio. He also showed that with $O(1/\epsilon)$ bits of advice per request it is possible to achieve a $1+\epsilon$ competitive ratio.\footnote{
Renault states the approximation as $1/(1-\epsilon)$ whereas we will use $(1 + \epsilon)$ which is justified since 
$1/(1-\epsilon) \leq 1+\epsilon$ for all $\epsilon \leq 2/3$.  Also without loss of generality, we will sometimes say that the approximation is $1+\Theta(\epsilon)$ since our advice bounds are asymptotic and we can replace $\epsilon$ by $\epsilon/c$ for some suitable $c$.} 
In \cite{renault2016online} Renault explicitly asked, as an open problem, to analyze the advice complexity of the dual bin packing problem in the tape advice model. In this paper, we make progress on the advice complexity needed for achieving a ($1+\epsilon$) competitive ratio for the online dual bin packing problem. Specifically, let $s$ be the maximum bit size of a weight of an input item. In particular, the overall input size is $O(ns)$ bits. We present an online algorithm that with $O(\frac{s+\log n}{\epsilon^2})$ bits of advice achieves a ($1+\epsilon$) competitive ratio for the dual bin packing problem. Note that it is trivial to achieve optimality with $n \log_2 m$ advice bits by specifying for each input item into which bin it should be placed. When stated in the tape advice model, Renault's bound for a ($1+\epsilon)$ competitive ratio is $\Theta(n/\epsilon)$. Our advice bound for achieving a $(1+\epsilon)$ approximation is exponentially smaller for the regime of constant $\epsilon$ and $s = O(\log n)$. When the $n$ item weights have $s=n$ bits of precision, we show that the dependence on $s$ is necessary by exhibiting an $\Omega_{\epsilon}(n)$ lower bound on the advice necessary to achieve a ($1+\epsilon$) approximation.   

Our main result heavily relies on a simple polynomial time approximation scheme (PTAS) for the dual bin packing problem, which constitutes the technical core of this paper. Dual bin packing is a special case of the multiple knapsack problem (MKP). In the MKP, each of the $n$ items is described by its weight (number in $(0,1]$) and its value (an integer). There are $m$ knapsacks each with their own capacity. The goal is to pack a subset of items such that all items fit into the knapsacks without violating weight constraints and the total value of packed items is as large as possible. In the uniform MKP, capacities of the bins are equal, and, are taken to be $1$ without loss of generality. Thus, the dual bin packing problem can be seen as the uniform MKP with all values being $1$. It is known \cite{ChekuriK05} that the dual bin packing, and consequently the MKP, is strongly NP-hard even for $m = 2$, which effectively rules out an FPTAS for these problems. This is in contrast to the standard knapsack problem and the makespan problem for a fixed number of machines where FPTAS are possible. Significant progress in the study of the MKP was made by Kellerer \cite{kellerer1999polynomial} who showed that the uniform MKP admits a PTAS. Subsequently, Chandra and Khanna \cite{ChekuriK05} gave a PTAS for the general MKP. Clearly, these results also give PTAS algorithms for the dual bin packing problem. However, the PTAS algorithms provided by Kellerer, and Chekuri and Khanna, are relatively complicated algorithms with a technically detailed analysis of correctness. Our goal is to provide a simple online advice algorithm for the dual bin packing problem based on a simpler PTAS for the dual bin packing problem . Thus, as a first step, we provide a simpler PTAS and analysis for the case of the dual bin packing problem. In the second step, we use the simplified PTAS to derive our result for the tape advice-complexity of the online dual bin packing. We use a dynamic programming algorithm to solve a restricted version of the dual bin packing problem instead of relying on IP solvers as in Kellerer's PTAS or the LP solver as in Chekuri and Khanna. This dynamic programming algorithm is essentially the same as the one used in the solution of the makespan problem with a bounded number of different processing times. Our PTAS and its analysis are self-contained and easy to follow. Our work highlights one of the important aspects of simple algorithms, namely, they are usually easier to modify and adapt to other problems and situations. In particular, we are able to easily adapt our simple PTAS to the setting of online tape-advice algorithms.

\section{Preliminaries}
\label{sec:prelim}
The dual bin packing instance is specified by a sequence of $n$ item weights $w_1, w_2, \ldots, w_n$ and $m \in \mathbb{N}$ bins, where $w_i \in (0,1]$. The goal is to pack a largest subset of items into $m$ bins such that for each bin the total weight of items placed in that bin is at most $1$. The problem can be specified as an integer program as follows (notation $[n]$ stands for $\{1,\ldots,n\}$):

\begin{align*}
\text{max. } & \sum_{i=1}^n \sum_{j=1}^m x_{ij}&\\
\text{subj. to } & \sum_{i=1}^n x_{ij}w_i \le 1 & \text{for all }j \in [m]\\
& \sum_{j=1}^m x_{ij} \le 1 & \text{for all }i \in[n]\\
& x_{ij} \in \{0,1\} & \text{for all }i \in [n], j \in [m]\\
\end{align*}

The online First Fit algorithm \textsc{FF} constructs a solution by processing items in the given order $w_1, w_2, \ldots, w_n$ and placing a given item into the first bin into which it fits. First Fit Increasing algorithm \textsc{FFI} first orders the items by increasing weight. Let $\sigma :[n] \rightarrow [n]$ be the corresponding permutation. Then \textsc{FFI} runs \textsc{FF} on the items in the order given by $\sigma$, i.e., $w_{\sigma(1)} \le w_{\sigma(2)} \le \cdots \le w_{\sigma(n)}$. It is easy to see that \textsc{FF} has an unbounded approximation (i.e., competitive)  ratio whereas Coffman et al \cite{CoffmanLT78} show that \textsc{FFI} has a $4/3$ approximation ratio for the dual packing problem. 
Note that the weights only enter the above integer program as constraints and are not part of the objective function. Thus, it is easy to see that the items in an optimal solution are, without loss of generality, a prefix of $w_{\sigma(1)}, w_{\sigma(2)}, \ldots, w_{\sigma(n)}$ packed into appropriate bins.  

Throughout this paper we shall always write $n$ to mean the number of input items, $m$ the number of bins, and $s$ the maximum bit size of an input item ($s$ can be thought of as the  ``word size'' of a computer, on which the given input sequence should be processed).

An online algorithm ALG is said to achieve a \emph{competitive ratio} $c$ for a maximization problem if there exists a constant $\alpha$ such that for all input sequences $I$ we have $\textsc{OPT}(I) \le c \textsc{ALG}(I) + \alpha$, where $\textsc{ALG}(I)$ is the value of the objective that the algorithm achieves on $I$ and $\textsc{OPT}(I)$ is the value achieved by an offline optimal solution. If $\alpha \le 0$, we say that ALG achieves a \emph{strict} competitive ratio $c$.

\section{A Simple PTAS for the Dual Bin Packing Problem}
\label{sec:ptas}

Fix $\epsilon > 0$. Let $S = \{ i \mid w_i \le \epsilon\}$ be the set of small input items, and let $L = \{ i \mid w_i > \epsilon\}$ be the set of large input items. The goal is to pack as many items from $S \cup L$ into $m$ bins as possible. Our first observation is that if the \textsc{FFI} algorithm fills $m$ bins (i.e., does not allow any more items to be packed) using only small items then it already achieves a $1+\epsilon$ approximation.

\begin{lemma}\label{lem:ffi-suffices}
Suppose that when the  \textsc{FFI} algorithm terminates,  it has filled all bins with items of weight at most $\epsilon$. Then \textsc{FFI} achieves $1+\epsilon$ approximation ratio on this instance.
\end{lemma}
\begin{proof}
If \textsc{FFI} packs all items then it clearly finds an optimal solution. Suppose that \textsc{FFI} rejects some items. Let $w$ be the smallest weight of a rejected item. Thus the total remaining free space among all $m$ bins is $< wm$ in the \textsc{FFI} packing. Thus, \textsc{OPT} can pack at most $m-1$ more items, since it can only add items of weight $\ge w$. Let $N$ be the number of items packed by \textsc{FFI}. Then we have
\[ \frac{\textsc{OPT}}{\textsc{FFI}} < \frac{N+m}{N} \le \frac{m/\epsilon+m}{m/\epsilon} = 1+\epsilon,\]
where the second inequality follows from $N \ge m/\epsilon$, since the \textsc{FFI} packing uses only items of weight $\le \epsilon$.
\end{proof}

Thus, the whole difficulty in designing a PTAS for this problem lies in the handling of large items. If \textsc{FFI} terminates before packing all of $S$, then 
the condition of Lemma \ref{lem:ffi-suffices} holds and hence from now on, we consider the case when \textsc{FFI} packs all of $S$. In this case an optimal solution is to pack all of $S$ together with some subset of smallest items from $L$. The strategy for our algorithm is to pack a largest subset $F$ of $L$ that still leaves enough room to pack all of $S$. This means that $w(F) \le m - w(S)$, but we also want to pack all of $S$ efficiently. This can be guaranteed by leaving slightly more room while packing $F$. Namely, $w(F) \le m(1-\epsilon)-w(S)$ guarantees that all of $S$ can be packed efficiently after packing $F$.

\begin{lemma}[Kellerer \cite{kellerer1999polynomial}]
Suppose that we have a packing of $F \subseteq L$ such that $w(F) \le m(1-\epsilon)-w(S)$. Then running \textsc{FFI} with the packing of $F$ as a starting point results in packing all of $S$.
\end{lemma}
\begin{proof}
Initially, we have $w(F) \le m(1-\epsilon) - w(S) \le m(1-\epsilon)$. Thus, by the pigeonhole principle there is a bin with $\ge \epsilon$ free space. Thus, we can pack the first item $s_1$ from $S$. Now, we have $w(F) \le m(1-\epsilon) - w(S) \le m(1-\epsilon) - w(s_1)$, i.e., $w(F) + w(s_1) \le m(1-\epsilon)$. Again, by the pigeonhole principle there is a bin with $\ge \epsilon$ free space, so we can pack the second item from $S$, and so on.
\end{proof}
\begin{remark} \label{rem:ff-completion}
Note that the argument in the above lemma does not use the \emph{increasing} property of \textsc{FFI}. Therefore, even \textsc{FF} can be used to complete the partial packing $F$ with all of $S$.
\end{remark}

The next lemma shows that the extra ``breathing room'' that we leave to guarantee an efficient packing of $S$ does not hurt the approximation ratio.

\begin{lemma}
Let $F$ be the largest subset of $L$ that can be packed into $m$ bins with total weight $\le m(1-\epsilon)-w(S)$. Then
\[ \frac{\textsc{OPT}}{|F|+|S|} \le 1+ 3 \epsilon.\]
\end{lemma} 
\begin{proof}
If $F=L$ then we are done. Otherwise, let $w > \epsilon$ be the smallest weight of an item from $L \setminus F$. Then $|S| \ge w(S)/\epsilon \ge w(S)/w$ and $|F| \ge \frac{m(1-\epsilon)-w(S)}{w}$. Thus, $|F|+|S| \ge \frac{m(1-\epsilon)}{w}$. The total free space after packing $F \cup S$ is $< \epsilon m$. Thus, \textsc{OPT} can pack at most $\frac{\epsilon m}{w}$ more items than $|F|+|S|$. Combining all of the above, we have
\[ \frac{\textsc{OPT}}{|F|+|S|} \le \frac{|F|+|S|+\epsilon m/w}{|F|+|S|} \le \frac{m(1-\epsilon)/w + \epsilon m/w}{m(1-\epsilon)/w}  = \frac{1}{1-\epsilon} \le 1+3 \epsilon,\]
where the last inequality holds for small epsilon; i.e., $\epsilon \le 2/3$.
\end{proof}

We shall refer to the problem of finding $F$ as in the above lemma as the LFP (``the large $F$ problem'').

\begin{remark} Suppose that $F$ is an approximation to the LFP with an additive $\epsilon m$ term, i.e. $|F| \ge \textsc{OPT}_{LFP} - \epsilon m$. Then an argument similar to the one used in the above lemma shows that $F$ together with $S$ still gives $1+\Theta(\epsilon)$ approximation to the original dual bin packing problem. Thus, it suffices to find a good enough $F$.
\end{remark}

Before we show how to find a good approximation to the LFP, we show how to solve the dual bin packing optimally in polynomial time when the number of distinct weights of the input items is fixed. As previously stated, this follows from the known PTAS for the makespan problem. (See section 10.2 of the Vazirani text \cite{vazirani-text}.)

\begin{lemma}
\label{lem:dp}
We can solve the dual bin packing problem optimally in time $O(n^{2k} m)$ where $k$ is the number of distinct weights of the input items.
\end{lemma}
\begin{proof}
The algorithm is a simple dynamic programming. Let $w_1, \ldots, w_k$ be the distinct weights appearing in the input. The entire input sequence can be described by a $k$-tuple $(n_1, \ldots, n_k)$, where $n_i$ is the number of items of weight $w_i$ and $n = \sum_i n_i$. Note that the number of different possible $k$-tuples with $n$ items is $O(n^k)$. Let $\mathcal{K}$ be the set of distinct $k$-tuples such that each of its element fits entirely in a single bin, i.e., $(\ell_1, \ldots, \ell_k)$ such that  $\sum_i \ell_i w_i \le 1$. The dynamic programming table $D$ is going to be indexed by the number of available bins $m'$ and a possible $k$-tuple $(\ell_1, \ldots, \ell_k)$ such that $0 \le \ell_i \le n_i$. The value $D[(\ell_1, \ldots, \ell_k), m']$ is going to indicate the maximum number of items that can be packed from the input sequence described by the state $(\ell_1, \ldots, \ell_k)$ in $m'$ bins. Let $\mathcal{L} = \{(\ell_1', \ldots, \ell_k') \mid \forall i~ 0\le \ell_i' \le \ell_i\}.$  An optimal solution to the subproblem indexed by $(\ell_1, \ldots, \ell_k)$ and $m'$ consists of a packing of some element from $\mathcal{L}$ into a single bin, and packing the remaining input items into $m'-1$ bins:
\[ D[(\ell_1, \ldots, \ell_k), m'] = \max_{(\ell_1', \ldots, \ell_k') \in \mathcal{K} \cap \mathcal{L}} \sum_{i} \ell_i' + D[(\ell_1 - \ell_1', \ldots, \ell_k - \ell_k'), m'-1].\]
The base case is given by the states where either $m'=0$ or $\sum_i \ell_i = 0$, in which case we cannot pack any items. The overall runtime of this algorithm is $O(n^{2k} m)$ since the dynamic programming table has $O(n^k m)$ entries and each entry can be computed in time $O(n^k)$ with appropriate preprocessing of the input data. As usual, this dynamic program can be easily modified to return the actual packing rather than the number of packed items.
\end{proof}

Let $L'$ be the subset of the smallest items from $L$ such that $|L'|$ is as large as possible subject to $w(L') \le m(1-\epsilon)-w(S)$. We would like to find $F$ by running the dynamic programming algorithm on $L'$. Unfortunately, $L'$ can have too many distinct inputs. The idea is to group items of $L'$ into few groups depending only on $\epsilon$, reassign all weights of elements within a single group to the weight of the largest element in that group, and run the dynamic programming algorithm on the new problem instance. Then, we will need to argue that the resulting solution is an additive $\epsilon m$ approximation to the LFP.

Let $\ell = |L'|$. We can assume $\ell > m$ otherwise there is a trivial way to pack $\ell$ items into $m$ bins. Assume for simplicity that $m \epsilon$ is an integer and that $k = \ell/(m\epsilon)$ is also an integer. Then, we split $L'$ into $k$ groups of $m\epsilon$ elements each. Let $w_{j_1} \le w_{j_2} \le \cdots \le w_{j_\ell}$ be the weights of elements in $L'$. Define $L_i$ to be the $i$th group consisting of items of weights $w_{j_{1+(i-1)m\epsilon}}, \ldots, w_{j_{i m \epsilon}}$. Reassign the weights of elements in $L_i$ to be $w_{j_{i m \epsilon}}$. Let $\widetilde{w}$ denote the modified weights. Thus, we get an instance with $k$ distinct weights, where $k = \ell/(m\epsilon)$. Note that $\ell \le m/\epsilon$ since we are dealing with large items, so $k \le 1/\epsilon^2$. Thus, we can solve this instance in time $O(n^{2/\epsilon^2}m)$ by Lemma~\ref{lem:dp}. Let $F'$ denote this solution. Let $F$ denote an optimal solution to LFP with the original weights. Then, we have the following.

\begin{lemma}
$F'$ is feasible with respect to weights $w$ and $|F'| \ge |F| - \epsilon m$.
\end{lemma}
\begin{proof}
Since $F'$ is feasible with weights $\widetilde{w}$ and $\widetilde{w} \ge w$, we immediately conclude that $F'$ is feasible with respect to $w$. Rather than directly showing $|F'| \ge |F| - \epsilon m$,  we show how to construct a set $F''$ from $F$ such that $|F''| \ge |F| - \epsilon m$ and $F''$ is feasible with respect to $\widetilde{w}$. This will prove the lemma, since $F'$ is a maximum cardinality set that satisfies the feasibility constraints (i.e., $|F'| \geq |F''|$ ). To construct such $F''$, we can simply drop all items from $F \cap L_1$ and replace all items from $F \cap L_i$ by arbitrary items from $L_{i-1}$ for $i \ge 2$. Note that $|F''| = | F \setminus L_1 | \ge |F| - \epsilon m$. Moreover, $F''ß$ is feasible with respect to $\widetilde{w}$ since we are always replacing large weight items by smaller weight items.
\end{proof}

This completes the argument that approximately solving the LFP using the reassigned weights and the dynamic programming followed by \textsc{FFI} on small items gives a $1+\Theta(\epsilon)$ approximation. The running time of the dynamic programming is $O(n^{2/\epsilon^2}m)$. One can run  \textsc{FFI} in  $O(n \log n +nm)$ time. The overall running time of our PTAS algorithm is $O(n^{2/\epsilon^2}m + n \log n)$, which is clearly polynomial when $\epsilon$ is fixed. Algorithm~\ref{algo:ptas} desribes this PTAS.

\begin{algorithm}[!h]
\caption{Our PTAS for the dual bin packing.}\label{algo:ptas}
\begin{algorithmic}
\Procedure{Dual Bin Packing PTAS}{$w_1 \le \cdots \le w_n, m, \epsilon$}
\State{Let $S = \{i \mid w_i \le \epsilon\}$}
\If{\textsc{FFI}($S,m$) packs $< |S|$ items}
\Return{\textsc{FFI}($S,m$)}
\EndIf
\State{Let $L' = \{|S|+1, \ldots, |S|+\ell\}$ be the indices of the  largest subset of $L$ such that $w(L') \le m(1-\epsilon)-w(S)$}
\State{Let $k = \ell/(m\epsilon)$}
\State{Let $\widetilde{w}$ denote new weights where items with indices $\{|S|+(i-1)m\epsilon+1, \ldots, |S|+im\epsilon\}$ all receive weight $w_{|S|+im\epsilon}$}
\State{Use the algorithm of Lemma~\ref{lem:dp} to obtain a packing of items $F'$ with modified weights $\widetilde{w}$}
\State{Regard $F'$ as a packing with the original weights}
\State{Run $FF$ on $S$ with the packing of $F'$ as a starting point} \\
\Return{the resulting packing}
\EndProcedure
\end{algorithmic}
\end{algorithm}


Summarizing, in this section we proved the following theorem.

\begin{theorem}\label{thm:ptas}
Algorithm~\ref{algo:ptas} is a PTAS for the dual bin packing problem.
\end{theorem}

\section{Advice Complexity of the Online Dual Bin Packing Problem for Bounded Bit Size of Input Items}
\label{sec:advice}
In this section, we consider the online version of the dual bin packing problem in the tape-advice model. Let $s$ be the maximum bit-size of an input item weight. Then the input bit-length is $O(sn)$. Based on the PTAS in Algorithm \ref{algo:ptas}, we develop an online algorithm that achieves $1+\epsilon$ approximation to the dual bin packing problem with $O\left( \frac{s + \log n}{\epsilon^2}\right)$ bits of advice. 
Before we prove the main result of this section, we need to modify Lemma~\ref{lem:ffi-suffices} to work in the online setting. Recall that Lemma~\ref{lem:ffi-suffices} detects when \textsc{FFI} is already successful enough that we don't need to do any extra work to obtain a $1+\epsilon$ approximation. An online algorithm does not have the ability to sort the input items, thus we would like to obtain a version of Lemma~\ref{lem:ffi-suffices} that detects when \textsc{FF} obtains a $1+\epsilon$ approximation. The restricted subsequence first fit (\textsc{RSFF}) algorithm given by Renault \cite{renault2016online} is what we need. Let $W = w_1, \ldots, w_n$ be the sequence of weights given to the online algorithm. For a value $\eta$ we define $W_\eta$ to be the subsequence $(w_i \mid w_i \le \eta)$. The \textsc{RSFF} algorithm finds the largest value of $\eta$ such that \textsc{FF} packs all items in $W_\eta$ and then returns \textsc{FF}$(W_{\eta})$. Without loss of generality, we may assume that $\eta$ is one of the $w_i$. 

\begin{lemma}[Implicit in Renault \cite{renault2016online}]
\label{lem:renault}
If \textsc{RSFF} identifies an $\eta$ such that $\eta \le \epsilon$ then \textsc{RSFF} achieves a $1+\epsilon$ approximation ratio.
\end{lemma}

By replacing \textsc{FFI} with \textsc{RSFF} in the first step of Algorithm~\ref{algo:ptas}, we obtain the main result of this section.

\begin{theorem}
There is an online algorithm achieving a $1+\Theta(\epsilon)$ strict competitive ratio for the dual bin packing problem with $O\left( \frac{s + \log n}{\epsilon^2}\right)$ bits of advice, where $s$ is the maximum bit-size of an input item. 
\end{theorem}
\begin{proof}
The advice is obtained by slightly modifying the PTAS from Section~\ref{sec:ptas}. At first, the oracle writes down the value of $\eta$ identified by running \textsc{RSFF} on the input sequence. For later convenience, we rename $\eta$ by $\widetilde{w}_0$. This takes $O(s)$ bits of advice. This is analogous to running \textsc{FFI} in the original PTAS. Recall that the PTAS creates $k \le 1/\epsilon^2$ groups of large input items $L_i$ for $i \in [k]$ with the corresponding rounded weights $\widetilde{w}_i$ for $i \in [k]$. The oracle appends $|L_i|$ together with $\widetilde{w}_i$ for $i \in [k]$ to the advice string. This completes the specification of the advice string. The length of the advice string is $O(s + k \log |L_i| + k s) = O\left( \frac{s + \log n}{\epsilon^2}\right)$.

It is left to see that with this advice string an online algorithm can computes a $1+\Theta(\epsilon)$ approximate solution to the instance of the dual bin packing problem. Observe that if $\widetilde{w}_0 \le \epsilon$ then by Lemma~\ref{lem:renault} the solution obtained by 
running \textsc{FF} on all items of weight $\le \widetilde{w}_0$  achieves $1+\epsilon$ approximation, since this gives us exactly the packing produced by \textsc{RSFF}. From now on, we consider the case $\widetilde{w}_0 > \epsilon$. Then an optimal solution might use large items. Recall that the PTAS creates a solution to the rounded instance encoded by $(|L_1|, \ldots, |L_k|)$ and weights $(\widetilde{w}_1, \ldots, \widetilde{w}_k)$, replaces this solution with actual weights of the corresponding items and fills the rest in \textsc{FF} fashion with the rest of the items (see Remark~\ref{rem:ff-completion}). Thus, knowing  $(|L_1|, \ldots, |L_k|)$ and weights $(\widetilde{w}_1, \ldots, \widetilde{w}_k)$ from the advice, our online algorithm can reserve place holders for items in bins according to the dynamic programming solution. We refer to this space as the preallocated space, and we refer to the complement of it as the remaining space. For example, if dynamic programming solution says that bin 1 contains $\ell_i$ items of weight $\widetilde{w}_i$ then the online algorithm reserves $\ell_i$ slots of weight $\widetilde{w}_i$ in bin 1. The preallocated space in bin 1 is $\sum_{i} \ell_i \widetilde{w}_i$ and the remaining space in bin 1 is $1 - \sum_{i}\ell_i \widetilde{w}_i$. Now, the algorithm is ready to process the items in the online fashion. When the algorithm receives an input item of weight $\le \epsilon$ it packs it in the remaining space in \textsc{FF} fashion. When the algorithm receives an item of weight $\in (\widetilde{w}_{i-1}, \widetilde{w}_i]$, it packs it into the first available preallocated slot of weight $\widetilde{w}_i$. By the construction of advice, we are guaranteed that when the algorithm is done processing the inputs, all preallocated slots are occupied and all small items are packed. By Theorem~\ref{thm:ptas} this solution is a $1+\epsilon$ approximation. 
\end{proof}

\section{Advice Complexity of the Online Dual Bin Packing Problem for General Weights}
\label{sec:advice-general}
In this section we show that the online dual bin packing without any restrictions on $s$ requires $\Theta_\epsilon(n)$ advice to approximate \textsc{OPT} within $1+\epsilon$. Observe that the upper bound $O(n/\epsilon)$ immediately follows from the result of Renault~\cite{renault2016online} in the per-request advice model. A somewhat stronger upper bound, $(1-\Omega(\epsilon))n$, follows by observing that the dual bin packing belongs to the advice complexity class AOC defined by Boyar et al.~\cite{boyar2016advice} and then using the results from ~\cite{boyar2016advice}. Thus, we only need to prove that in the case of unrestricted $s$ the lower bound of $\Omega_\epsilon(n)$ holds. For sufficiently small $\epsilon$, we show a nearly matching lower bound of $(1-O(\epsilon \log (1/\epsilon)))n = \Omega_\epsilon(n)$ in the tape-advice model. We establish our lower bound by providing a reduction (that preserves the precision, advice and competitive ratio) from an online problem known to require a lot of advice to the dual bin packing problem. The starting point is the binary separation problem defined by Boyar et al.~\cite{boyar2016binpacking}.

\begin{definition}[Boyar et al.~\cite{boyar2016binpacking}]
The \emph{binary separation problem} is the online problem with input $I=(n_1, y_1, \ldots, y_n)$ consisting of $n = n_1 + n_2$ positive values which are revealed one by one. There is a fixed partitioning of the set of items into a subset of $n_1$ large items and a subset of $n_2$ small items, so that all large items are greater than all small items. Upon receiving an item $y_i$, an online algorithm for the problem must guess if $y$ belongs to the set of small or large items. After the algorithm has made a guess, it is revealed whether the guess was correct. The goal is to maximize the number of correct guesses.
\end{definition}

Boyar et al.~\cite{boyar2016binpacking} establish a lower bound on the advice needed to achieve competitive ratio $c$ for the binary separation problem.

\begin{theorem}[Boyar et al.~\cite{boyar2016binpacking}]
\label{thm:bin-sep-lb}
Assume that an online algorithm solves the binary separation problem on sequences $I=(n_1, y_1, \ldots, y_n)$ where the $y_i$ are $n$ bit numbers and does so using at most $b(n)$ bits of advice while making at most $r(n)$ mistakes. Set $\alpha = (n-r(n))/n$. If $\alpha \in [1/2,1)$ then $b(n) \ge (1-H(\alpha))n$ 
where $H(p) = p \log (1/p) + (1-p) \log (1/(1-p))$.
\end{theorem}

Moreover, Boyar et al.~\cite{boyar2016binpacking} provide a reduction from the binary separation problem to the standard bin packing problem to show that achieving competitive ratio $< 9/8$ requires an online algorithm to receive $\Omega(n)$ bits of advice. A simple adaptation of this reduction allows us to derive a similar result for the dual bin packing problem. We present the details below for completeness.

\begin{theorem}
An online algorithm achieving a competitive ratio $1+\epsilon$ for the dual bin packing problem with unrestricted bit size of input weights requires $(1-O(\epsilon \log (1/\epsilon)))n = \Omega_\epsilon(n)$ bits of advice, provided $\epsilon < 1/19$.
\end{theorem}
\begin{proof}
We show how to reduce the binary separation problem to the dual bin packing problem while preserving the size of advice and the competitive ratio.

Let ALG be an algorithm for the dual bin packing problem that achieves competitive ratio $c$ and uses advice $b(n)$. Let $I=(n_1, (y_1, \ldots, y_n))$ be an input to the binary separation problem. We define $ALG'$ for solving $I$ as follows. $ALG'$ constructs an instance of the dual bin packing problem in the online fashion. It will use decisions and the advice string of ALG to make decisions about its own inputs $y_i$. Let $\delta_{\max}>\delta_{\min} > 0$ be small enough numbers. Suppose that we have a strictly decreasing function $f: \mathbb{R}\rightarrow(\delta_{\min},\delta_{\max})$. $ALG'$ invokes ALG with $n$ bins and $2n$ items. $ALG'$ constructs input weights to ALG in three phases.
\begin{description}
\item[Phase 1 (preprocessing):] the first $n_1$ weights are defined as $1/2+\delta_{\min}$. This is generated by $ALG'$ prior to any inputs seen from $I$.
\item[Phase 2 (online):] when $y_i$ arrives, $ALG'$ defines a new input item to ALG of weight $1/2-f(y_i)$. In this phase $ALG'$ uses decisions of ALG to handle its own inputs. If ALG packs the current item into a bin that contains $1/2+\delta_{\min}$ item from phase 1 then $ALG'$ declares $y_i$ to be large. Otherwise, $ALG'$ declares the item to be small. We shall refer to $1/2-f(y_i)$ weight items corresponding to truly small (large) $y_i$ as small items (large items). 
\item[Phase 3 (post processing):] once $ALG'$ has processed the entire sequence $I$, it appends weights $1/2+f(y_i)$ for all truly small $y_i$ from $I$. We refer to these weights as the complementary weights of small items.
\end{description}
First observe that \textsc{OPT} for the constructed instance of the dual bin packing packs all $2n$ items into $n$ bins: the $n_1$ weights corresponding to the large items can be paired up with the $n_1$ items from phase 1, and the $n_2$ weights corresponding to the small items can be paired up with their complementary weights from phase 3 in the remaining $n_2 = n - n_1$ bins.

Clearly, the advice complexity and the precision of the input items are preserved by this reduction. Thus, to finish the argument we need to analyze how many mistakes $ALG'$ does. We bound the number of mistakes in terms of the number of items unpacked by ALG. We define the following variables.
\begin{itemize}
\item Let $p_1$ be the number of items from phase 1 that were not packed by ALG.
\item Let $\ell_2$ be the number of large items from phase 2 that were not packed by ALG.
\item Let $s_2$ be the number of small items from phase 2 that were not packed by ALG.
\item Let $p_3$ be the number of items from phase 3 that were not packed by ALG.
\end{itemize}

The overall number of items that were not packed by ALG is $p_1+\ell_2+s_2+p_3 \le \frac{c-1}{c} 2n$. 
Observe that the complementary weights can only be paired up with the corresponding small item weights, and phase 1 items can only be paired up with large or small phase 2 items.

The number of bins containing phase 1 items is $n_1 - p_1$. The number of bins containing phase 3 items is $n_2-p_3$. Due to the above observations, all these bins have to be distinct. Thus,  the number of bins that do not contain either phase 1 or phase 3 items is $p_1+p_3$. Call these the leftover bins. The are two types of mistakes that $ALG'$ can do: (1) it classifies a large item as being small, and (2) it classifies a small item as being large. Since large items can only be paired either with phase 1 items or be placed in the leftover bins, type (1) mistakes occur only when large items are placed in the leftover bins or when large items remain unpacked. There can be at most $2(p_1+p_3)$ large items in the leftover bins. Thus, $ALG'$ makes at least $g_1:=n_1 - p_1 - 2(p_1+p_3)-\ell_2$ correct guesses for large items. A type (2) mistake happens only when a small item is paired up with a phase 1 item. Since there can be at most $n_1-p_1-g_1 = 2(p_1+p_3)+\ell_2$ phase 1 items not paired up with large items, there can be at most that many type (2) mistakes. Thus, $ALG'$ makes at least $g_2 = n_2 - s_2 - 2(p_1+p_3)-\ell_2$ correct guesses for small items. Overall, $ALG'$ makes $g_1+g_2 = n_1 + n_2 - s_2-p_1-4(p_1+p_3)-2\ell_2 \ge n_1+n_2-5(p_1+\ell_2+s_2+p_3) \ge n - 10\frac{c-1}{c}n$ good guesses. The fraction of good guesses is then $\frac{n-10(c-1)n/c}{n} = \frac{10-9c}{c}$. By Theorem~\ref{thm:bin-sep-lb}, it follows that $b(n) \ge (1-H((10-9c)/c)) n$. Observe that $(10-9c)/c \in (1/2, 1)$ provided that $c \in (1, 20/19)$. In particular, if $\epsilon$ is a small positive constant, then achieving a competitive ratio $c = 1+\epsilon$ for the dual bin packing problem requires $(1-H((1-9\epsilon)/(1+\epsilon)))n = (1-H(O(\epsilon)) ) n = (1-O(\epsilon \log (1/\epsilon)))n=\Omega_\epsilon(n)$ bits of advice.
%
\end{proof}

All in all, the dual bin packing problem admits short advice in case of $s$ bounded by a slowly growing function of $n$, but requires long advice when $s$ is unrestricted. This is akin to the distinction between the polynomial time vs pseudo-polynomial time in the regular Turing machine world. One of the conceptual contributions of this paper is a demonstration that ``pseudo-short'' advice and ``truly short'' advice are provably different. To make this idea precise, we introduce two natural classes of efficient advice problems.

\begin{definition}
The class \emph{EAC (efficient advice complexity)} consists of online problems $P$ such that an input to $P$ is given by $n$ items, and the advice complexity of achieving $1+\epsilon$ competitive ratio for $P$ is $O_\epsilon(\poly(\log n))$.

Denoting the maximum bit size of an input item to $P$ by $s$, we define a superclass \emph{WEAC (weakly efficient advice complexity)} of EAC to consist of those online problems $P$ such that the advice complexity of achieving $1+\epsilon$ competitive ratio for $P$ is $O_\epsilon(\poly(\log n, s))$.
\end{definition}

EAC class is defined by analogy with communication complexity where $O(\poly(\log n))$ communication is considered efficient (see Babai et al.~\cite{babai1986}). WEAC class is also natural. The advice length bound of algorithms for WEAC problems suggests that the advice can consist of a short description of combinatorial parameters of a problem (e.g., length of a stream, index into a stream, which take $O(\log n)$ bits to describe) plus a small (polylogarithmic) number of actual data items from the stream.

In light of the above definitions and the main result of this section and Section~\ref{sec:advice}, the dual bin packing problem witnesses the following class separation theorem.

\begin{theorem}
WEAC$\neq$EAC.
\end{theorem}

\section{Conclusion}
\label{sec:conclusion}
We presented a simple PTAS for the dual bin packing problem. Although a PTAS for a more general multiple knapsack problem was already known, our PTAS is arguably simpler to state and analyze. Its simplicity helped us to adapt it to the tape-advice model of online algorithms. We showed that a $1+\epsilon$ competitive ratio for the dual bin packing problem is achievable with $O\left( \frac{s+\log n}{\epsilon^2} \right)$ bits of tape advice. We showed that the dependence on $s$ is necessary to obtain such small advice, as the dual bin packing problem  requires $\Omega_\epsilon(n)$ when $\epsilon >0$ is small enough, $s$ is unrestricted, and $m$ is part of the input. 
We introduced two natural advice complexity classes EAC and WEAC. The conceptual distinction between the classes WEAC and EAC is similar to the Turing machine world distinction between pseudo-polynomial time and strongly polynomial time. EAC captures problems that can be approximated to within $1+\epsilon$ with $O_\epsilon(\poly \log n)$ bits of advice, whereas WEAC captures problems that can be approximated to within $1+\epsilon$ with $O_\epsilon(\poly(\log n, s))$ bits of advice. Our results on the dual bin packing problem imply that WEAC$\neq$EAC. 

One immediate question left open by our work is whether there is an small advice  algorithm for small $s$ which requires less advice bits. More specifically, does there exist a $1+ \epsilon$ approximation using $o_{\epsilon}(s) + O_\epsilon(\log n)$ advice bits for $s = o(n)$? In this paper we exclusively studied the dual bin packing in the regime of obtaining $1+\epsilon$ competitive ratio when $\epsilon$ is small and $m$ is part of the input. Are there sublinear advice algorithms for large $\epsilon$, e.g., $\epsilon = 1/2$? Also, does the dual bin packing admit sublinear advice algorithms when $m$ is a small constant? It is also interesting to see whether or not results for the dual bin packing problem can be extended to more general problems such as  when bins have different capacities, and more generally to the  multiple knapsack problem, while preserving conceptual simplicity. Last and perhaps a most important question is whether or not there 
exist online algorithms with efficiently computable (i.e., linear or even online computable as in 
\cite{Durr2016,borodinps17}) advice for the dual bin packing problem achieving a constant competitive ratio.

\bibliography{ptas}{}
\bibliographystyle{plain}

\end{document}